\documentclass[preprint,review,12pt,authoryear]{elsarticle}

\usepackage{geometry}
\usepackage{times}

\usepackage{amsmath,amssymb,amsthm}
\newtheorem{theorem}{Theorem}
\newtheorem{proposition}{Proposition}
\newtheorem{lemma}{Lemma}

\theoremstyle{definition}

\newtheorem{assumption}{Assumption}

\usepackage{graphicx}
\usepackage{booktabs}
\usepackage{array}
\usepackage{float}
\usepackage{comment}
\usepackage{pgfplots}
\pgfplotsset{width=16cm,compat=1.18}
\usetikzlibrary{arrows.meta}

\usepackage{caption}

\usepackage{hyperref}
\hypersetup{
    colorlinks=true,
    linkcolor=blue,
    urlcolor=blue,
    citecolor=blue
}
\journal{Economics Letters}

\begin{document}

\begin{frontmatter}

\title{\textbf {Entry deterrence and antibiotic conservation under post-entry Bertrand competition}}

\author[HofstraEcon]{Roberto Mazzoleni\corref{cor1}}
\ead{roberto.mazzoleni@hofstra.edu}

\author[HofstraMath]{Hamza Virk}
\ead{hvirk2@pride.hofstra.edu}

\cortext[cor1]{Corresponding author}

\address[HofstraEcon]{Department of Economics, Hofstra University, Hempstead, NY 11549, USA}
\address[HofstraMath]{Department of Mathematics, Hofstra University, Hempstead, NY 11549, USA}

\begin{abstract}
We analyze how an incumbent antibiotic monopolist responds to the threat of post-entry Bertrand competition by a vertically differentiated rival.  In a two-period model where current production drives future resistance, Bertrand competition leads to a winner-take-all outcome. We find that strategic deterrence is optimal regardless of bacterial cross-resistance to prospective rival drugs. In contrast with post-entry Cournot competition, anticipated price competition provides the incumbent with a stronger strategic incentive for conservation.
\end{abstract}


\end{frontmatter}

\section{\small{Introduction}}
Antibiotic resistance has been framed as a tragedy of the commons \citep{Hardin1968,Tisdell1982}. Accordingly, the commons of antibiotic effectiveness are degraded by the combined effects of excessive use on the demand side, and the misalignment between society's interest in conserving effective therapies for bacterial infections and drug manufacturers' interest in profiting from innovative antibiotics before generics and rival drugs reach the market. 

If consumption of novel antibiotics exceeds socially optimal levels, proposals advocating for the extension of the duration and breadth of patent protection on antibiotic drugs aim to weaken the patent-holder's incentive to dissipate antibiotic effectiveness \citep{Kades2005,Gonzalez2022}. Such policies have well known drawbacks \citep{Laxminarayan2002, OUTTERSON2007,EswaranGallini2019} . Indeed, stronger protection for newly patented antibiotic drug influences negatively the incentive to develop future ones \citep{GilbertShapiro1990,Klemperer1990, Scotchmer1991} 

This is particularly important in light of the declining rate of discovery of new antibiotics and the phenomenon of bacterial cross-resistance and cross-sensitivity. When exposure to one antibiotic leads to the development of bacterial cross-resistance to others, drugs draw from the same pool of antibiotic effectiveness \citep{Laxminarayan2002}. But bacterial cross-resistance among antibiotics varies and, in some cases, bacterial cross-sensitivity may evolve, driving recurrent interest in drug cycling protocols \citep{ImamovicSommer2013,BaymStoneKishony2016}. 

These interactions have strategic implications for firms' decisions on horizontal and vertical differentiation, and entry deterrence \citep{HungSchmitt1988, DonnenfeldWeber1992, NohMoschini2006, KaraerErhun2015}. For instance, \citet{HerrmannNkuiyaDussault2013} analyzes how bacterial cross-resistance influences the innovation strategy of a firm attempting to develop a new antibiotic. Greater differentiation from existing generic antibiotics is costly but would reduce the extent to which a novel drug draws from the same pool of antibiotic effectiveness as the existing generic drug. \citet{Mazzoleni2025} focuses on how a patent-holding incumbent firm will respond to the threat of future entry by rival drug subject to varying degrees of bacterial cross-resistance. An important finding is that, depending on development costs and predicted bacterial cross-resistance to rival drugs, optimal deterrence requires conserving a drug's effectiveness rather than dissipating it. 

We extend the model in  \citet{Mazzoleni2025} by analyzing an incumbent firm's strategy facing the threat of entry and Bertrand competition rather than Cournot competition as in that model. We demonstrate two results.  First, strategic deterrence is optimal under broader conditions when post-entry Bertrand competition is expected. Second, and most significantly, optimal deterrence of a Bertrand competitor requires that the incumbent conserve antibiotic drug efficacy by producing less before entry than what a monopolist would produce in the absence of competitive entry threats.  In contrast, optimal deterrence under Cournot competition requires either conserving or depleting antibiotic effectiveness depending on the level of development costs and the degree of bacterial cross-resistance to the entrant drug.  

In our two-period model, first period output decisions by the incumbent affect competitive conditions of the post-entry duopoly by determining the effectiveness gap benefiting prospective rivals and the overall size of the second period market demand. Under Bertrand competition, the second period duopolistic market can never be profitable for the incumbent. In particular, a higher quality rival would successfully monopolize the market rather than share it with the incumbent, so that the incumbent has a stronger incentive to deter entry. This goal can be attained by reducing either the effectiveness gap separating the incumbent from the rival drug or the market size. Our model shows that optimal deterrence always requires reducing first-period output and conserving the incumbent drug's efficacy.

The following section presents the model and our main theoretical results. We compare them with those in \citet{Mazzoleni2025} in Section 3 and comment briefly on the policy implications of our work.

\section{\small {Optimal deterrence under varying conditions of entry costs and bacterial cross-resistance}}
We adopt the two-period model framework of \citet{Mazzoleni2025}.  Accordingly, an incumbent (I) operates as a monopolist and chooses output $X$ during the first period. During period 2, a potential entrant (E) observes $X$ and decides whether or not to incur a fixed entry cost $R \geq 0$ to participate in Bertrand competition with the incumbent.  If entry occurs, firms compete in prices (Bertrand).  We characterize the strategic behavior of the incumbent at the subgame perfect Nash Equilibrium of our model.

We assume a linear demand system consistent with the  Quasilinear Quadratic Utility Model (\citet{ChoneLinnemer2020}), so the first period's inverse demand is $P(X) = \alpha - \beta X$. First-period use generates resistance, so the second-period inverse demand functions are:
\begin{align}
    p_I(x, y, X) &= (\alpha - \theta X) - \beta(x+y) \label{eq:demand_I}\\
    p_E(x, y, X) &= (\alpha - \phi X) - \beta(x+y) \label{eq:demand_E}
\end{align}
where $x$ and $y$ are the outputs of the incumbent and entrant, $\theta$ is own-resistance, and $\phi$ is cross-resistance. This structure implies the antibiotics are perfect substitutes but vertically differentiated by effectiveness. Both firms produce their drugs at constant and identical marginal cost, $c$. Our analysis throughout this section builds on the following assumptions.

\begin{assumption}
We assume: $0<\theta < 2\beta$.
\label{assump:markup}
\end{assumption}
This assumption ensures a positive primary mark-up ($\alpha - \theta X - c$) in the second-period market for the incumbent drug should output in the first period be set at the static profit-maximizing level $X=(\alpha-c)/2 \beta$.  Thus, it would never be optimal for the incumbent to extract all profits in the first period and abandon the market in the second period.

\begin{assumption}
We assume $\theta \geq \phi \geq 0$. 
\label{assump:resistance}
\end{assumption}
Own-resistance is at least as great as cross-resistance. If $\theta = \phi$, the two drugs are equivalent as if the entrant firm's drug were a generic version of the incumbent's. If $\theta > \phi$, the entrant has a quality advantage over the incumbent and the effectiveness gap $\Delta(X) = (\theta - \phi)X$ increases linearly with the incumbent's first-period output.

\subsection{Monopoly Benchmark}

As a benchmark, we note that a two-period monopolist facing no entry threat would choose the output of the first period to maximize the sum of profits of the first and second periods: $\Pi^M(X)=\Pi_1(X) + \pi_2^M(X)$, where $\Pi_1(X) = (\alpha-c-\beta X)X$, and $\pi_2^M(X) = \frac{1}{4\beta}(\alpha-c - \theta X)^2$. Under Assumption \ref{assump:markup}, $\Pi^{M}(X)$ is strictly concave. The profit-maximizing output is $X_M^* = (\alpha-c)/(2\beta+\theta)$. The monopolist restricts output below the static optimum ($X_S^* = (\alpha-c)/2\beta$) to conserve future effectiveness.

\subsection{Second-Period Bertrand Competition}

We begin characterizing the subgame perfect Nash equilibrium of the game by studying equilibrium conditions of the second-period Bertrand competition subgame in the event that entry occurs. Given the demand equations  (\ref{eq:demand_I}) and (\ref{eq:demand_E}), if both firms are to produce and sell in the market, their prices must differ by exactly the quality differential: $p_E=p_I+\Delta(X)$. Either of the two will wish to undercut the rival until either prices fall to marginal cost or the entrant monopolizes the market.

\begin{lemma}
\
\begin{enumerate}

    \item If either $\theta=\phi$ or if  $\theta>\phi$ and $X=0$ (Generic Entry), then $p_{I}^{*}=p_{E}^{*}=c$ .
    \item If $\theta>\phi$ and $X>0$ (Differentiated Entry), then the entrant monopolizes the market and sets prices as follows:    \[
    p_{E}^{*}=
    \begin{cases}
        c+\Delta(X), & \text{if } X< \frac{\alpha-c}{2\theta-\phi} \\
        \\
        \frac{\alpha-\phi X+c}{2}, & \text{if } X\ge  \frac{\alpha-c}{2\theta-\phi}
    \end{cases}
    \]
\end{enumerate}
\label{lemma:pricing}
\end{lemma}

\begin{proof}
\textit{Case 1:}  Products are homogeneous ($\Delta(X)=0$), so prices reflect marginal cost as in standard Bertrand competition: $p_{I}^{*}=p_{E}^{*}=c$.

\textit{Case 2:} A positive effectiveness gap ($\Delta(X)>0$) ensures that the entrant can drive the incumbent out of the market by setting prices at the lower between the limit price $p_{E}^{L}=c+\Delta(X)$ and its monopoly price $p_{E}^{M}=(\alpha-\phi X+c)/2$. The limit price increases in $X$ while the monopoly price decreases. The two expressions are equal when $X=(\alpha-c)/(2\theta-\phi)$, which we will henceforth refer to as $\tilde{X}$.
\end{proof}

Crucially, the incumbent earns zero profit post-entry ($\pi_{I}^{*}(X)=0$) in either of the two cases.

\subsection{Entry Decision}

The entrant enters if their expected profit exceeds the entry cost $R$.

\subsubsection{Case 1: \texorpdfstring{($\theta=\phi$)}{(theta=phi)} }

Generic entry is never profitable under Bertrand competition if $R>0$, so the incumbent can maximize profits by setting  $X=X_M^*$. If $R=0$, then entry may or may not occur as net profits are zero. Predicting entry, the incumbent would choose to maximize profits in the first period by choosing $X=X_S^*$. Otherwise, $X=X_M^*$ would be the optimal output choice. Thus, overproduction during the first period ($X_S^*>X_M^*$) could only occur in exceptional circumstances. 

\subsubsection{Case 2: \texorpdfstring{($\theta > \phi$)}{(theta > phi)}}

Unless $X=0$, the entrant captures the market. The entrant's gross profit is $\pi_{E}^{*}(X)=(p_{E}^{*}-c)y^{*}(X)$:

\begin{align}
\pi_{E}^{*}(X) =
\begin{cases}
    \pi_{1}(X) \equiv \frac{(\theta-\phi)X(\alpha-c-\theta X)}{\beta}, & \text{if } X <\tilde{X} \\
    \\
    \pi_{2}(X) \equiv \frac{(\alpha-c-\phi X)^{2}}{4\beta}, & \text{if } X \ge \tilde{X}\\ 
\end{cases}
\label{eq:profitE}
\end{align}
where $\tilde{X}=(\alpha-c)/(2\theta -\phi)$. When $X=0$, the quality differential vanishes and Bertrand competition results in zero profit as in the case of generic entry.

\begin{lemma}
The entrant's profit function $\pi_{E}^{*}(X)$ is continuous and continuously differentiable. It attains its maximum value at: (a) $X_{\pi E}^{*}=\frac{\alpha-c}{2\theta}$ if $\phi>0$ ; or (b) $X_{\pi E}^* \ge \frac{\alpha -c}{2 \theta} =\tilde{X}$ if $\phi=0$.
\label{lemma:entrant_profit}
\end{lemma}

\begin{proof}
It can be verified that $\pi_{1}(\tilde{X}) = \pi_{2}(\tilde{X})$ and $\pi_{1}'(\tilde{X}) = \pi_{2}'(\tilde{X})$, so $\pi_{E}^{*}(X)$ is continuous and continuously differentiable. The first regime, $\pi_{1}(X)$, is a strictly concave quadratic function of $X$ attaining its maximum at $X_{\pi E}^{*}=(\alpha-c)/2\theta$. The second regime, $\pi_{2}(X)$, is a convex, monotonically decreasing function of $X$ if $\phi>0$. Since $\theta>\phi>0$, we have $2\theta > 2\theta-\phi$, and thus $X_{\pi E}^{*}=(\alpha -c)/(2\theta) < \tilde{X}$ is the unique maximum of the profit function. If $\phi=0$ , $\pi_{2}(X)$ is constant over the range $X \ge \tilde{X}$. Moreover, $(\alpha -c)/2\theta = \tilde{X}$ so that the profit function attains its maximum for any $X \ge \tilde{X}$ .
\end{proof}

\noindent
Entry occurs if $\pi_{E}^{*}(X)>R$ . Assuming: 
\begin{equation}
R \le \max\ {\pi_E^*(X)}=\frac{1}{\beta}\frac{\theta -\phi}{\theta}\left(\frac{\alpha-c}{2}\right)^2
\end{equation}

\noindent
we define an accommodation region $\mathbb{X}_{A}(R)$ and a deterrence region $\mathbb{X}_{D}(R)$. If $\phi>0$, $\mathbb{X}_{A}(R)=(X_{L}(R),X_{H}(R))$, where $X_{L}(R)$ and $X_{H}(R)$ are the roots of $\pi_{E}^{*}(X)=R$. If $\phi=0$, the accommodation region is $(X_L(R), \infty)$. If $\phi>0$, $\mathbb{X}_{D}(R)=[0, X_L(R)] \cup [X_H(R), (\alpha-c)/\phi]$. If $\phi=0$, $\mathbb{X}_{D}=[0,X_L (R)]$.

\subsection{Incumbent's Optimal Strategy}

Anticipating the rival's strategy, the incumbent chooses whether to accommodate or deter entry such as to maximize its own total profits.

\textbf{Profit under Accommodation ($X\in\mathbb{X}_{A}(R)$):} In this case, the incumbent would earn zero profit post-entry. Accordingly, incumbent's maximum profits under accommodation are defined as:
\begin{equation}
\Pi_{A}^*(R)=\max_{X \in {\mathbb{X_A}(R)}} \ \Pi_A(X)= \Pi_1(X) = (\alpha-c-\beta X)X
\label{eq:profitA}
\end{equation}
\textbf{Profit under Deterrence ($X\in\mathbb{X}_{D}(R)$):} The incumbent remains a monopolist during the second period, so it earns second period profits equal to $\pi_2^M(X) = (\alpha-c - \theta X)^2 /4\beta$. Overall profits under deterrence are given by expression $\Pi ^M(X)$ above, and we define
\begin{equation}
\Pi_{D}^*(R)= \max_{X \in \mathbb{X}_D(R)} \Pi^{M}(X)
\label{eq:profitD}
\end{equation}
The following theorem establishes that deterrence (or blockade) will be the incumbent's optimal strategy when entry costs are positive. Only when rival entry is costless, the incumbent will be indifferent between deterrence and accommodation.

\begin{theorem}
Under Bertrand competition with differentiated entry ($\theta>\phi$), $\Pi_{D}^*(R) \ge \Pi_{A}^*(R)$ whenever $R \ge 0$. 
\label{thm:deterrence}
\end{theorem}

\begin{proof}
Define the difference function $\Delta(R)=\Pi_{D}^{*}(R)-\Pi_{A}^{*}(R)$. We first establish that $\Delta(R)=0$ for $R=0$. 

If $R=0$, the accommodation region is the set where $\pi_E^*(X)>0$. Thus, $\mathbb{X}_{A}(0)=(0,(\alpha-c)/\phi)$ if $\phi>0$, or $\mathbb{X}_{A}(0)=(0, \infty)$ if $\phi=0$. By Assumption \ref{assump:resistance}, the static optimum $X_{S}^{*}=(\alpha-c)/2\beta$ lies within this interval. The maximum accommodation profit is $\Pi_{A}^{*}(R)|_{R=0}=\Pi_{A}(X_{S}^{*})=(\alpha-c)^{2}/4\beta$. The deterrence region is the set $\mathbb{X}_{D}(0)= [0] \cup [(\alpha -c)/\phi, \infty)$ when $\phi>0$ and $\mathbb{X}_{D}(0)= [0]$ when $\phi=0$. Either way $\Pi_D(X)$ attains its maximum at $X=0$. Since $\Pi_D^*(R)|_{R=0}=(\alpha -c)^2/4\beta$, we conclude that $\Delta(0)=0$. When entry is costless, the incumbent is indifferent between accommodation and deterrence.

Next, we show that $\Delta(R)>0$ if $R>0$. First, note that as $R$ increases, $\mathbb{X}_{A}(R)$ shrinks, so $\Pi_{A}^{*}(R)$ is weakly decreasing in $R$. Conversely, as $R$ increases, $\mathbb{X}_{D}(R)$ expands, implying $\Pi_{D}^{*}(R)$ is weakly increasing in $R$. In particular, it is strictly increasing at $X=0$ since:
\[
\frac{\partial{\Pi_M}}{\partial{X}}\bigg\rvert_{X=0}=(2\beta -\theta)(\alpha-c)>0.
\]
Thus, deterrence is profit-maximizing when entry costs are positive.
\end{proof}

\subsection{Output Choice for Optimal Deterrence}

Since deterrence is always preferred (Theorem \ref{thm:deterrence}), if entry is not blockaded at $X_{M}^{*}$, the incumbent must distort its output to the boundary of $\mathbb{X}_{D}(R)$ closest to $X_{M}^{*}$.

\begin{theorem} \label{thm:conservation}
Under Bertrand competition with $\theta > \phi$ and Assumption~\ref{assump:markup} ($0 < \theta < 2\beta$), if strategic deterrence is required (i.e., $X_M^* \notin \mathbb{X}_D(R)$), the incumbent's deterrence-optimal first-period output is
\[
X^* = X_L(R).
\]
\end{theorem}

\begin{proof}
The two-period monopoly objective $\Pi^M(X) = \Pi_1(X) + \pi_2^M(X)$ is a strictly concave quadratic with unique maximizer
\[
X_M^* = \frac{\alpha - c}{2\beta + \theta}.
\]
Hence there exists $K > 0$ such that
\[
\Pi^M(X) = \Pi^M(X_M^*) - K(X - X_M^*)^2,
\]
so among any two feasible $X$, the one closer to $X_M^*$ yields the larger $\Pi^M$.

\textit{Case 1: $\boldsymbol{\phi = 0}$.}
Here $\mathbb{X}_D(R) = [0, X_L(R)]$. Strategic deterrence requires $X_L(R) < X_M^*$, and since $\Pi^M$ is increasing on $[0, X_M^*]$, the best choice in $[0, X_L(R)]$ is $X_L(R)$.

\textit{Case 2: $\boldsymbol{\phi > 0}$}.
Now $\mathbb{X}_D(R) = [0, X_L(R)] \cup [X_H(R), (\alpha - c)/\phi]$. Thus, optimal deterrence is achieved by choosing whichever output level between $X_L(R)$ and $X_H(R)$ is closest to $X_M^*$. Our claim amounts then to showing that:
\[
X_M^*-X_L(R) <  X_H(R)-X_M^*
\]
or, equivalently, that:
\[
(X_M^*-X_{\pi E}^*)+(X_{\pi E}^*-X_L(R)) <  (X_H(R)-X_{\pi E}^*)+(X_{\pi E}^*-X_M^*)
\]
This is true since $(X_M^*-X_{\pi E}^*) < 0$ and $(X_{\pi E}^*-X_L(R))\leq  (X_H(R)-X_{\pi E}^*)$.
\end{proof}

\subsection{Equilibrium Characterization}

We summarize the results of our analysis and provide a characterization of the incumbent's strategy at the subgame perfect Nash equilibrium of our model of post-entry Bertrand competition.

\begin{proposition}
Let $\overline{R}_{B}$ be the entry cost such that $X_{L}(\overline{R}_{B})=X_{M}^{*}$. At the subgame perfect Nash equilibrium, entry by a vertically differentiated drug will be:
\begin{enumerate}
    \item blockaded ($X^{*}=X_{M}^{*}$) if  $\overline{R}_{B}\le R$. 
    \item deterred ($X^{*}=X_{L}(R)<X_{M}^{*}$) if  $0<R<\overline{R}_{B}$. 
    \item accommodated ($X^*=X_S^*>X_M^*$) if $R=0$. 
\end{enumerate}
\noindent
Deterrence promotes conservation of the incumbent drug's effectiveness.  Faster dissipation occurs only when entry by a more effective antibiotic is costless.
\label{prop:equilibrium}
\end{proposition}
Figure \ref{fig:deterrence} illustrates the case at (2.) for the following parameter values: $ \alpha=10,\beta=2, \theta=2, c=2, \phi=1/2$.
\bigskip
\begin{figure}[H]
    \begin{tikzpicture}
    
    \begin{axis}[
    axis lines = left,
    xmin=0,
    xmax=5,
    ymin=0,
    ymax=11,
    xlabel = {$X$},
    ylabel = {Profit},
    ytick={},
    xtick = {1.2, 3.35},
    xticklabels = {$X_L(R)=X^*$, $X_H(R)$},
    xticklabel style={yshift=-0.2cm},
    legend style={at={(0.98,0.98)}, anchor=north east, draw=none, fill=white, font=\small},
    legend cell align={left},
    width=16cm,
    height=12cm,
]
    \addplot [line width=1, domain=0:{8/3.5}, samples=100, color=red] 
                {(1.5*x*(8 - 2*x)) / 2};
    \addlegendentry{$\pi_E(X)$}
    
    \addplot[line width=2, domain=0:1.2, samples=100]{(8-2*x)*(1+0.75*x)};
        \addlegendentry{$\Pi_D^*$ }
    \addplot[line width=2, domain=1.2:3.35, samples=100, color=black!50]{(8-2*x)*x};
    \addlegendentry{$\Pi_A^*$ }
 
    \addplot[line width=1, domain=0:5, samples=100, color=red!75, dashed]{5};
    \addlegendentry{ $R$}
    
    \addplot[line width=1, domain={8/3.5}:5, samples=100, color=red]
                {(8 - 0.5*x)^2 /8};
    \addplot[line width=2, domain=3.35:4, samples=100, color=black!100]{(8-2*x)*(1+0.75*x)};
    
    \addplot[line width=0.8, domain=0:5, samples=100, color=black, dotted]{(8-2*x)*(1+0.75*x)};
   
    \addplot[line width=0.8, domain=0:5, samples=100, color=black!75, dotted]{(8-2*x)*x};
        
    \addplot[only marks, mark=*, mark size=2, color=red, fill=red] coordinates {(1.2,5) (3.35,5)};
    \addplot[only marks, mark=*, mark size=2, fill=black] coordinates {(1.2,10.65) (2,8)(3.35, 4.5)};
    \addplot[only marks, mark=*, mark size=2, draw=black, fill=white, line width=1] coordinates {(1.2,0) (3.35,0)};
    
    \draw[densely dotted, black!70, line width=1.2] (axis cs:1.2,10.7) -- (axis cs:1.2,0);
    \draw[densely dotted, black!70, line width=1.2] (axis cs:3.35,4.5) -- (axis cs:3.35,0);
    
    \node[anchor=north, font=\normalsize] at (axis cs:1,-1) {$X_L(R)=X^*$};
    \node[anchor=north, font=\normalsize] at (axis cs:4,-0.6) {$X_H(R)$};
    
    \node[anchor=south, font=\small, fill=white] at (axis cs:0.6,0.75) {Deterrence region};
    \draw[ultra thick, black!75, <->] (axis cs:0,0.5) -- (axis cs:1.2,0.5);
      \node[anchor=south, font=\small, fill=white] at (axis cs:4.15,0.75) {Deterrence region};
    \draw[ultra thick, black!75, <->] (axis cs:3.35,0.5) -- (axis cs:5,0.5);
    
    \node[anchor=south, font=\small, fill=white] at (axis cs:2.25,0.75) {Accommodation region};
    \draw[ultra thick, red!75, <->] (axis cs:1.2,0.5) -- (axis cs:3.35,0.5);
    
\end{axis}
\end{tikzpicture}
\captionsetup{justification=centering, font=small}
\caption{Optimal deterrence with underproduction}
\label{fig:deterrence}
\end{figure}
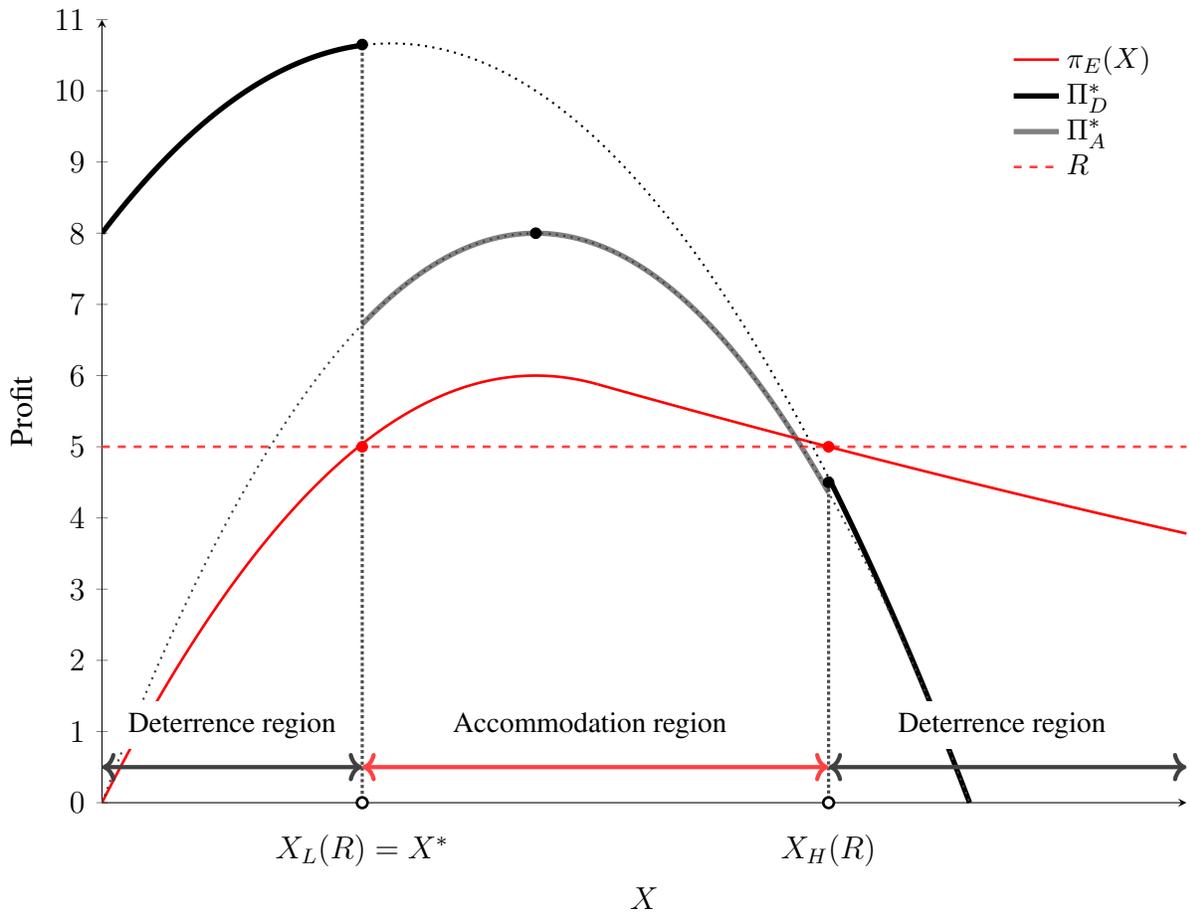

\bigskip
\section{Conclusions}
We extend the results in \citet{Mazzoleni2025} exploring conditions under which the threat of rival entry promotes strategic deterrence through antibiotic conservation. When Bertrand competition prevails post-entry rather than Cournot, strategic deterrence is an optimal response over a broader set of conditions on (a) the rival drug's development cost, and (b) the degree of bacterial cross-resistance to the prospective rival. Whereas in the case of Cournot competition deterrence calls for either dissipating or conserving the incumbent drug's effectiveness, Bertrand competition makes the effectiveness gap between incumbent and rival drugs a crucial strategic concern.  Accordingly, deterrence is achieved by conserving the incumbent drug's effectiveness and limiting the prospective quality gap with the rival drug.

Our model and \citet{Mazzoleni2025} show that predicted patterns of market competition play a significant role in the management of antibiotic effectiveness. The significance of these findings is heightened by empirical evidence challenging the widespread belief that anticipation of generic entry accelerates the development of bacterial resistance.  For instance, \citet{Kållbergetal2021}'s study of prescription volumes and average prices for antibiotics in the US found no consistent pattern across a sample of drugs for which generic entry occurred at patent expiration. Furthermore, generics for off-patent drugs, including antibiotics, are not always available in the market \citep{Guptaetal2018}.  Further theoretical exploration of the strategic implications of bacterial cross-resistance (and cross-sensitivity) and patterns of post-entry competition will contribute to better understand what policy tools can address effectively the antibiotic resistance crisis. 

\section*{Funding sources}
This research did not receive any specific grant from funding agencies in the public, commercial, or not-for-profit sectors.
\bibliographystyle{elsarticle-harv}
\bibliography{EcoLet_References}

\end{document}